\documentclass[preprint,12pt]{elsarticle}

\usepackage{amssymb}

\usepackage{amsmath}

\usepackage{booktabs}
\biboptions{numbers,sort&compress}
\newdefinition{remark}{Remark}
\newtheorem{theorem}{Theorem}
\newdefinition{assumption}{Assumption}
\newproof{proof}{Proof}
\newtheorem{lemma}{Lemma}
\newdefinition{example}{Example}
\newdefinition{definition}{Definition}
\newcommand{\lV}{\left\Vert}
\newcommand{\rV}{\right\Vert}

\usepackage{algorithm}
\usepackage{algpseudocode}
\usepackage{color}
\allowdisplaybreaks

\journal{Systems \& Control Letters}

\begin{document}
	
	\begin{frontmatter}
		
		
		
		\title{Privacy-Preserving Distributed Online Mirror Descent for Nonconvex Optimization}
		
		
		\author[1]{Yingjie Zhou} 
		\ead{Zhouyingjie7@163.com}
		\author[2]{Tao Li\corref{mycorrespondingauthor}}
		\cortext[mycorrespondingauthor]{Corresponding author.}
		\ead{tli@math.ecnu.edu.cn}
		
		\address[1]{School of Mathematical Sciences, East China Normal University, Shanghai, 200241, China}
		\address[2]{Key Laboratory of Mathematics and Engineering Applications, Ministry of Education, School of Mathematical Sciences, East China Normal University, Shanghai, 200241, China}
		\begin{abstract}
			We investigate the distributed online nonconvex optimization problem with differential privacy over time-varying networks. Each node minimizes the sum of several nonconvex functions while preserving the node's differential privacy.
			We propose a privacy-preserving distributed online mirror descent algorithm for nonconvex optimization, which uses the mirror descent to update decision variables and the Laplace differential privacy mechanism to protect privacy. Unlike the existing works, the proposed algorithm allows the cost functions to be nonconvex, which is more applicable. Based upon these, we prove that if the communication network is $B$-strongly connected and the constraint set is compact, then by choosing the step size properly, the algorithm guarantees $\epsilon$-differential privacy at each time. Furthermore, we prove that if the local cost functions are $\beta$-smooth, then the regret over time horizon $T$ grows sublinearly while preserving differential privacy, with an upper bound $O(\sqrt{T})$.  Finally, the effectiveness of the algorithm is demonstrated through numerical simulations.
		\end{abstract}
		
		
		
		\begin{keyword}
			Nonconvex problems \sep differential privacy \sep distributed online optimization \sep
			regret analysis.
			
			
			
		\end{keyword}
		
	\end{frontmatter}
	
		\section{Introduction}
		In distributed optimization, each node is aware of its local cost function. Nodes exchange information with their neighboring nodes through a communication network and collaborate to optimize the global cost function. This collaboration results in convergence towards the global optimal solution, where the global cost function is the sum of the local cost functions of all nodes. In numerous practical scenarios, each node's local cost function varies over time which leads to distributed online optimization problems.
		Distributed online optimization problems have wide applications in signal processing (\cite{application_signal1}), economic dispatch (\cite{application_energy1}), and sensor networks (\cite{application_decision}) et al. The regret is commonly used to evaluate the performance of a distributed online optimization algorithm. If the regret grows sublinearly, then the algorithm is effective. The existing works show that if the cost function is convex, then the best regret bound achieved by an algorithm is $O(\sqrt{T})$, and if the cost function is strongly convex, then the best regret bound is $O(\ln(T))$, where $T$ denotes the time horizon (\cite{optimal_regret_bound,bestregret}). In recent years, distributed online optimization problems have been extensively studied (\cite{online_ADMM, distributed_online1,time-varying-graph,pseudoconvex}).
		The algorithms above are designed based on the Euclidean distance. Such algorithms often face challenges in computing projections for complex cost functions and constraint sets, e.g. problems with simplex constraints. Beck et al. (\cite{beck2003mirror}) proposed the mirror descent algorithm based on the Bregman distance instead of the Euclidean distance, which effectively improves the computational efficiency. The mirror descent algorithm is highly effective for handling large-scale optimization problems (\cite{ben2001ordered}), and numerous studies have focused on extending this algorithm to distributed settings (\cite{distributed_mirror_descent1,distributed_mirror_descent2}).

			In distributed algorithms, each node possesses a data set containing its private information. By collecting the information exchanged between nodes, attackers could potentially deduce the data sets of the nodes, leading to privacy leakage. To protect private information, an effective approach is the differential privacy mechanism proposed by Dwork et al. (\cite{differential_privacy}). The fundamental principle of the differential privacy mechanism is to add a certain level of noises to the communication, thereby ensuring that attackers can only obtain limited privacy information from their observations.  Recently, extensive works have been proposed for distributed online convex optimization problems with differential privacy (\cite{hongyiguang_stronglyconvex,sconvex_umbalanced_graph,convex_directed_graph,time-varying_differential_privacy,stronglyconvex,online_defferential_privacy_extra1,online_defferential_privacy_extra2}). Yuan et al. (\cite{hongyiguang_stronglyconvex}) proposed a privacy-preserving distributed online optimization algorithm based on the mirror descent. For a fixed graph, they proved that if the cost function is strongly convex, then the regret bound of the algorithm is the same as that without differential privacy, which is $O(\ln{(T)})$. Zhao et al. (\cite{sconvex_umbalanced_graph}) proposed a distributed online optimization algorithm with differential privacy based on one-point residual feedback. They proved that, even if the gradient information of the cost function is unknown, the algorithm can still achieve differential privacy and the regret of the algorithm grows sublinearly. For constrained optimization problems, Lü et al. (\cite{convex_directed_graph}) introduced an efficient privacy-preserving distributed online dual averaging algorithm by using the gradient rescaling strategy. Li et al. (\cite{time-varying_differential_privacy}) proposed a framework for differentially private distributed optimization. For time-varying communication graphs, they obtained the regrets  bounds $O(\sqrt{T})$ and $O(\ln{(T)})$ for convex and strongly convex cost functions, respectively. Zhu et al. (\cite{stronglyconvex}) proposed a differentially private distributed stochastic subgradient online convex optimization algorithm based on weight balancing over time-varying directed networks. They demonstrated that the algorithm ensures $\epsilon$-differential privacy and provided the regret bound of the algorithm.
			
			The aforementioned privacy-preserving distributed online optimization algorithms all address the cases of convex cost functions. However, in practical applications, many problems involve nonconvex cost functions. For example, in machine learning, the cost function may be nonconvex due to sparsity and low-rank constraints (\cite{nonconvex_example_1}). In wireless communication, energy efficiency problems may also involve nonconvex cost functions due to constraints on node transmission power (\cite{nonconvex_example_2}). Since the cost function is nonconvex, finding a global minimizer is challenging. Therefore, the regret for convex optimization cannot be used to measure the performance of nonconvex online algorithms. For nonconvex optimization, the goal is usually to find a point that satisfies the first-order optimality condition (\cite{Hazan_nonconvex,online_Newton,online_nonconvex,lu_nonconvex1}).
			They used the regret based on the first-order optimality condition to evaluate the performance of the algorithm.
			Under appropriate assumptions, they proved that the regret grows sublinearly.
			
			The privacy-preserving distributed online nonconvex optimization problems have wide applications in practical engineering fields, such as distributed target tracking.
			Up to now, the studies on privacy-preserving distributed nonconvex optimization have been restricted to offline settings (\cite{wang2023decentralized,khajenejad2022guaranteed}). In \cite{wang2023decentralized}, a new algorithm for decentralized nonconvex optimization was proposed, which can enable both rigorous differential privacy and convergence.	
Compared with the existing works, the main challenges for privacy-preserving distributed online nonconvex optimization  are as follows. The first challenge is ensuring the convergence of the algorithm with the differential privacy mechanism. Existing works on distributed online nonconvex optimization do not incorporate privacy protection, and adding noises to decision variables may lead to divergence for differential privacy in distributed online optimization algorithms. The second challenge is improving the computational efficiency of online algorithms. For online optimizations, computational resources are limited, and the immediate decision-making is required. Compared with \cite{wang2023decentralized}, which addressed unconstrained optimization problems, we focus on optimization problems with set constraints, making the solution process more difficult and the computational cost higher. The algorithm in \cite{wang2023decentralized} is based on the gradient descent. For constrained optimization problems, directly incorporating Euclidean projection into the algorithm given in \cite{wang2023decentralized} may lead to high computational costs. Motivated by selecting Bregman projections based on different constraint conditions can improve computational efficiency, we combine the differential privacy mechanism with the distributed mirror descent method to design a privacy-preserving distributed online  mirror descent algorithm for nonconvex optimization (DPDO-NC).  The proposed algorithm can address distributed online nonconvex optimization problems while protecting privacy. To evaluate the algorithm's performance, we use the regret based on the first-order optimality condition and give a thorough analysis of the algorithm's regret. Table \ref{tab1} compares our algorithm with existing privacy-preserving distributed online optimization algorithms. The main contributions of this paper are summarized as follows.
				\begin{table}[h]
					\centering
					\caption{Comparison of privacy-preserving distributed online algorithms.}\label{tab1}
					\begin{tabular}{c|c|c|c}
						\hline
						Reference & Cost functions & Constraint & Communication graph \\ \hline
						\cite{hongyiguang_stronglyconvex} & strongly convex & convex set & fixed undirected \\
						\cite{sconvex_umbalanced_graph}, \cite{online_defferential_privacy_extra1} & convex/strongly convex &  convex set & fixed directed \\
						\cite{convex_directed_graph} & convex & convex set  & fixed directed \\
						\cite{time-varying_differential_privacy} & convex/strongly convex &convex set &$\!$time-varying undirected \\
						\cite{stronglyconvex} & convex/strongly convex & no & time-varying directed \\
						\cite{online_defferential_privacy_extra2} & convex & equality & fixed undirected\\
						our work  & nonconvex & convex set& time-varying directed \\ \hline
					\end{tabular}
				\end{table}


In our algorithm, each node adds noises to its decision variables and then broadcasts the perturbed variables to neighboring nodes via the communication network to achieve consensus. Finally, the mirror descent method is used to update the local decision variables. We prove that the proposed algorithm maintains $\epsilon$-differential privacy at each iteration. Compared with \cite{lu_nonconvex1}, our algorithm adds Laplace noises during communication with neighbors to protect privacy information, and the communication topology in our algorithm is a time-varying graph rather than a fixed graph. Compared with \cite{differential_privacy,hongyiguang_stronglyconvex,sconvex_umbalanced_graph,convex_directed_graph,time-varying_differential_privacy}, we do not require the cost function to be convex. While \cite{differential_privacy,hongyiguang_stronglyconvex,sconvex_umbalanced_graph} used fixed graphs, we consider time-varying graphs, which have broader application scenarios.
		
		To overcome the challenges posed by constraints and online nonconvex cost functions, the proposed algorithm updates decisions by solving subproblems. The subproblems include the Bregman divergence. If the Bregman divergence satisfies appropriate conditions, then the algorithm can effectively track its stationary point. By using the mirror descent algorithm, we can choose different Bregman divergences depending on the problem, which can reduce the computational cost. By appropriately selecting the step size, we prove that the regret upper bound of the algorithm is $O(\sqrt{T})$, i.e., the algorithm converges to the stationary point sublinearly. The regret order is the same as the regret for distributed online nonconvex optimization algorithms without differential privacy.
		
		The numerical simulations demonstrate the effectiveness of the proposed algorithm. By using the distributed localization problem, we demonstrate that the algorithm's regret is sublinear. Additionally, the numerical simulations reveal the tradeoff between the level of privacy protection and the algorithm's performance.

		The main structure and content of this paper are arranged as follows. Section 2 introduces the preliminaries, including the graph theory and the problem description. Section 3 presents the differentially private distributed online nonconvex optimization algorithm. Section 4 includes the differential privacy analysis and the regret analysis. Section 5 gives the numerical simulations. The final section concludes the paper.
		
		Notations: let $\mathbb{R}$ and $\mathbb{Z}$ be the set of real numbers and integers, respectively. For a vector or matrix $X$, $\Vert X \Vert$ and $\Vert X \Vert_1$ represent its 2-norm and 1-norm, respectively. Let $X^T$ be the transpose of $X$, and $[X]_i$ be the $i$-th row of the matrix. For a given function $f$, $\nabla f$ denotes its gradient. Let $\textbf{1}$ be a column vector of appropriate dimension with all elements equal to 1. Let $I$ be an identity matrix of appropriate dimension.
		For a random variable $x$, $\mathbb{P}[x]$ and $\mathbb{E}[x]$ denote its probability distribution and expectation, respectively. For a random variable $\xi$, $\xi \sim Lap(\sigma)$ indicates that $\xi$ follows a Laplace distribution with scale parameter $\sigma$, whose probability density function is given by $f(x|\sigma)=\frac{1}{2\sigma}e^{-\frac{|x|}{\sigma}}$.

		\section{Preliminaries and Problem Description}
		This section provides the problem description and preliminaries for the subsequent analysis.
		
		\subsection{Graph theory}
		The communication network is modeled by a directed graph (digraph) $G(t)=\{V, E(t), A(t)\}, t=1, 2, \dots$, where $V=\{1,2,\dots,N\}$ is the set of nodes and $E(t)$ is the set of edges. Elements in $E(t)$ are represented as pairs $(j, i)$, where $j, i \in V$, and $(j, i) \in E(t)$ indicates that node $j$ can directly send information to node $i$. Let $N_i^{in}(t)=\{j|(j,i)\in E(t)\}\cup \{i\}$ and $N_i^{out}(t)=\{j|(i,j)\in E(t)\}\cup \{i\}$ denote the in-neighbors and out-neighbors of node $i$ at each time $t$, respectively.
		The matrix $A(t)=\{a_{ij}(t)\}_{N \times N}$ is called the weight matrix of the graph $G(t)$. If $(j,i)\in E(t)$, then $a < a_{ij}(t) <1$ for some $0<a<1$, and $a_{ij}(t) = 0$ otherwise.
		If $A(t)$ is symmetric, then $G(t)$ is an undirected graph.
		If $A(t)$ satisfies $\boldsymbol{1}^T A(t) = \boldsymbol{1}^T$ and $A^T(t) \boldsymbol{1} = \boldsymbol{1}$, then $A(t)$ is called a doubly stochastic matrix, and $G(t)$ is called a balanced graph.
		For a fixed graph $G=\{V, E, A\}$, if for any two nodes $i, j \in V$, there exist $k_1, k_2, \dots, k_m \in V$ such that $(i, k_1), (k_1, k_2), \dots, (k_m, j) \in E$, then the graph $G$ is strongly connected. For the digraph $G(t)$, a $B$-edge set is defined as $E_B (t) = \bigcup_{k=(t-1)B+1}^{tB} E(k)$ for some integer $B \geq 1$. The digraph $G(t)$ is $B$-strongly connected if the digraph with the node set $V$ and the edge set $E_B(t)$ is strongly connected for all $t\geq1$ (\cite{stochastic_matrix}). 
		
		\subsection{Problem description}
		In this paper, we consider the time-varying system with $N$ nodes, where each node only knows its cost function $f_t^i$ and communicates with neighboring nodes through a directed time-varying network $G(t)$. The nodes cooperate to solve the following optimization problem
		\begin{equation}\label{p1}
			\begin{gathered}
				\min _{x\in \Omega } \sum_{i=1}^N f_t^i(x), \ t=1,\dots,T,
			\end{gathered}
		\end{equation}
		where $f_t^i: \mathbb{R}^d \rightarrow \mathbb{R}$ is the cost function of node $i$ at time $t$, and $T$ is the time horizon. The following basic assumptions are given for problem (\ref{p1}).

		\begin{assumption}\label{AS1}
			For all $t=1,\dots, T$, $G(t)$ is $B$-strongly connected, and the weight matrix $A(t)$ is doubly stochastic.
		\end{assumption}

		\begin{assumption}\label{AS2}
			The set $\Omega \subseteq \mathbb{R}^d$ is a bounded closed convex set.
		\end{assumption}
		
		\begin{assumption}\label{AS3}
			For all $t=1,\dots,T$, $i=1,\dots,N$, $f_t^i(x)$ is differentiable and $\beta$-smooth with respect to $x$, i.e., there exist $\beta>0$ such that
			\begin{align}\nonumber
				\lV\nabla f_t^i(x)-\nabla f_t^i(y)\rV\leq\beta\lV x-y\rV, \quad  \forall \, x,y\in \mathbb{R}^d.
			\end{align}
		\end{assumption}
		
		By Assumptions \ref{AS2}-\ref{AS3}, we know that there are $\eta > 0$ and $\theta > 0$ such that $\|x_1 - x_2\| \leq \eta$ for all $x_1, x_2 \in \Omega$, and $\|\nabla f_t^i(x)\| \leq \theta$ for all $i=1, \dots, N$, $t=1, \dots, T$, and $x \in \Omega$.
		Assumptions \ref{AS2}-\ref{AS3} are common in the research on distributed optimization (\cite{time-varying_differential_privacy}\cite{online_defferential_privacy_extra1}). Moreover, there are no convexity requirements for the cost functions $f_t^i$ in this paper. We give the following practical example that satisfies Assumptions \ref{AS1}-\ref{AS3}.
		
		\begin{example}\label{nonconvexexample}
			Distributed localization problem (\cite{online_Newton}\cite{example11}). Consider $N$ sensors collaborating to locate a moving target. Let $x_t^0 \in \Omega \subseteq \mathbb{R}^d$ be the true position of the target at time $t$, and $s_i$ be the position of sensor $i$. Each sensor can only obtain the distance measurement between the target's position and its position, given by
			$$d_t^i = \|s_i -x_t^0\| +\vartheta_t^i, \,\,\, i=1,\dots, N,$$
			where $\vartheta_t^i$ is the measurement error at time $t$. To estimate the target's position, the sensors collaborate to solve the following optimization problem
			$$\min_{x\in \Omega } \sum_{i=1}^N \frac{1}{2}\left| \|s_i -x\| - d_t^i \right|^2,$$
			where $\Omega$ is a bounded closed convex set representing the range of the target's position. It can be observed that $f_t^i(x)=\frac{1}{2}\left| \|s_i -x\| - d_t^i \right|^2$ is nonconvex with respect to $x$.
		\end{example}
		
		In distributed online convex optimization, each node uses local information and information transmitted from neighboring nodes to make a local prediction $x_{t+1}^i$ of the optimal solution $x^*$ at time $t+1$, where $x_{t+1}^i$ is the decision variable of node $i$ at time $t+1$.
		Node $i$ can only observe $f_t^i(x_t^i)$ at time $t$, so in order to solve the global optimization task, the nodes must communicate with each other through the network $G(t)$.

Generally speaking,
		find the global optimal point $x^*$ is impossible for nonconvex optimization.
We use the following definition of regret.
		
		\begin{definition}\label{regret_new}
			(\cite{hua2024distributed}) Let $\{x_t^i, \, t=1,\dots,T, \, i=1,\dots,N\}$ be the sequence of decisions generated by a given distributed online optimization algorithm, the regret of node $i$ is defined by
			\begin{equation}\label{expected_regret}
				\mathbb{E}[{R}_T^i] \triangleq  \max_{x\in {\Omega}} \left(\mathbb{E}\left[  \sum_{t=1}^T \sum_{j=1}^N  \langle \nabla f_t^j(x_t^i), x_t^i-x\rangle \right] \right) .
			\end{equation}
		\end{definition}
		
	 Different from the convex optimization, which aims at finding a global minimum, using the regret (\ref{expected_regret}) aims at finding a stationary point of a nonconvex cost function. Our goal is to design an efficient online distributed optimization algorithm such that the regret (\ref{expected_regret}) grows sublinearly with respect to time horizon $T$, i.e., $\lim_{T\rightarrow \infty} \mathbb{E}[{R}_T^i] /T = 0$, while also ensuring the privacy of each node's information.
		
\begin{remark}
Lu et al. (\cite{lu_nonconvex1}) provides the following individual regret for distributed online nonconvex optimization
\begin{equation}\label{regret}
			{R}_T^i \triangleq \max_{x\in \Omega} \left( \sum_{t=1}^T\sum_{j=1}^N \langle \nabla f_t^j(x_t^i), x_t^i-x\rangle \right),
\end{equation}
where $i=1,\dots,N$. If $x^*$ is a stationary point of $\sum_{t=1}^T \!\sum_{j=1}^N f_t^j(x)$, then $\sum_{t=1}^T \! \sum_{j=1}^N	\langle \nabla f_t^j(x^*), x^* - x \rangle \leq 0,$ $\forall \, x \in \Omega$.

Definition \ref{regret_new} is similar to the above, both describing the degree of violation of the first-order optimality condition.
The regret (\ref{regret}) is for the deterministic case, whereas in our algorithm, the noises are added, and therefore the regret (\ref{expected_regret}) is in the sense of expectation.

The regret (\ref{expected_regret}) is also applicable for distributed online convex optimization (\cite{time-varying_differential_privacy,stronglyconvex}), for which the regret is defined as
			\begin{equation}\label{convex_regret}
				\sum_{t=1}^T\sum_{j=1}^N \mathbb{E}[f_t^j \left(x_t^i\right)]- \sum_{t=1}^T\sum_{j=1}^N f_t^j \left(x^*\right),
			\end{equation}
			where $x^*=\arg\min_{x\in \Omega} \sum_{t=1}^T \sum_{j=1}^N f_t^j \left(x\right)$. By the properties of convex functions, we can obtain
			\begin{align*}
				\sum_{t=1}^T\sum_{j=1}^N \mathbb{E}[f_t^j \left(x_t^i\right)]- \sum_{t=1}^T\sum_{j=1}^N f_t^j \left(x^*\right) \leq \mathbb{E} \left[  \sum_{t=1}^T \sum_{j=1}^N  \langle \nabla f_t^j(x_t^i), x_t^i-x^*\rangle \right] \! \leq \! \mathbb{E}[{R}_T^i].
			\end{align*}	
From above, we can see that for convex functions, the regret (\ref{convex_regret}) grows sublinearly if  (\ref{expected_regret}) grows sublinearly.

\end{remark}
		
		\section{Algorithm Design}
		
		Without revealing the private information of individual nodes, in this paper, we propose a differentially private distributed online optimization algorithm for solving the nonconvex optimization problem (\ref{p1}), based on the distributed mirror descent algorithm and the differential privacy mechanism. The mirror descent algorithm is based on the Bregman divergence.
		
		\begin{definition}\label{bregman_definition}
			(\cite{distributed_mirror_descent1}) Given a $\omega$-strongly convex function $\varphi: \mathbb{R}^d \rightarrow\mathbb{R}$, i.e., there is  $\omega >0$ such that
			$$\varphi \left(y\right)-\varphi\left(x\right)\geq\nabla \varphi\left(x\right)^T\left(y-x\right)+\frac{\omega}{2}\Vert x-y\Vert^2, \, \forall \, x, y \in \mathbb{R}^d.$$
			The Bregman divergence generated by $\varphi$ is defined as
			\begin{equation}\label{bregman}
				D_{\varphi}(x,y)= \varphi(x)- \varphi(y)- \langle \nabla \varphi(y), x-y\rangle.
			\end{equation}
		\end{definition}
		
		The Bregman divergence measures the distortion or loss resulting from approximating $y$ by $x$. From Definition \ref{bregman_definition}, the Bregman divergence has the following properties: (i) $D_{\varphi}(x,y) \geq 0$. (ii) $D_{\varphi}(x,y)$ is strongly convex with respect to $x$.	We give some assumptions on the Bregman divergence.
		\begin{assumption}\label{AS4}
			(i) For any $x\in \Omega$, $D_{\varphi}(x,y)$ is convex respect to $y$, i.e., for any $r_i \geq 0, i=1,\dots, N,$ satisfying $\sum_{i=1}^N r_i= 1$,
			$$D_{\varphi}\left(x, \sum_{i=1}^N r_iy_i\right) \leq \sum_{i=1}^N r_i D_{\varphi}\left(x, y_i\right), \,\, \forall \, y_i \in \mathbb{R}^d.$$
			(ii) For any $x\in \Omega$, the Bregman divergence $D_{\varphi}(x,y)$ is M-smooth respect to $y$, i.e., there is $M>0$ such that
			$$D_{\varphi}(x,y_1)-D_{\varphi}(x,y_2) \leq \langle \nabla_y D_{\varphi}(x,y_2), y_1-y_2 \rangle + \frac{M}{2}\|y_1-y_2\|^2, \,\, \forall \, y_1, y_2 \in \mathbb{R}^d.$$
		\end{assumption}
		
		Assumption \ref{AS4} is a common assumption in the distributed mirror descent algorithm and is also mentioned in \cite{distributed_mirror_descent1} and \cite{hongyiguang_stronglyconvex}.
		
		Suppose an attacker can eavesdrop on the communication channel and intercept messages exchanged between nodes, which may lead to privacy leakage. To protect privacy, we add noises to the communication. At time $t$, each node $i$ adds noises to perturb the decision variables before communication. Then, each node communicates the perturbed decision variables with its neighbors for consensus. The specific iteration process is given by
		\begin{equation}\label{algorithm1}
			z_{t}^i =\sum_{j=1}^N a_{ij}(t) \left(x_t^j + \xi_t^j\right),
		\end{equation}
		where $\xi_t^j \sim Lap(\sigma_t)$. Next, each node $i$ updates its local decision variable using the mirror descent based on the consensus variable $z_t^i$. The specific iteration process is given by
		\begin{equation}\label{algorithm2}
			x_{t+1}^i = \arg \min_{x\in \Omega} \left\{ D_{\varphi}(x,z_t^i)+ \langle \alpha_t \nabla f_t^i(x_t^i), x \rangle \right\},
		\end{equation}
		where $\alpha_t$ is the step size, and $D_{\varphi}(x,z_t^i)$ is given by (\ref{bregman}). In (\ref{algorithm2}), $D_{\varphi}(x,z_t^i)+ \langle \alpha_t \nabla f_t^i(x_t^i), x \rangle$ is strongly convex with respect to $x$, which ensures a unique solution for $x_{t+1}^i$. The pseudocode for the algorithm is given as follows.
		
		\begin{algorithm}[H]
			\begin{algorithmic}[1] \label{algo1}
				\caption{Differentially private distributed online mirror descent algorithm for nonconvex optimization (DPDO-NC).}
				\State Input: step size $\alpha_t=\frac{1}{N\sqrt{t}}$, privacy level $\epsilon$, noise magnitude $\sigma_t =\frac{2\sqrt{d}\alpha_t\theta}{\omega \epsilon}$, Bregman divergence $D_{\varphi}(x,y)$, initial value $\{x_1^i\}_{i=1}^N \in \Omega$, weighted matrix $A(t)$.
				\For{$t = 1, 2, \dots, T$}
				\For{$i = 1, 2, \dots, N$ }
				\State Disturb $x_t^i$ by adding noises $\xi_t^i \sim Lap(\sigma_t)$, and broadcast $q_t^j=x_t^j + \xi_t^j$ to neighboring nodes.
				\State $z_{t}^i =\sum_{j=1}^N a_{ij}(t) q_t^j$.
				\State Update $x_{t+1}^i = \arg \min_{x\in \Omega} \left\{ D_{\varphi}(x,z_t^i)+ \langle \alpha_t \nabla f_t^i(x_t^i), x \rangle \right\}.$
				\EndFor
				\EndFor
				\State Output: $\{x_{T}^i\}_{i=1}^N$.
			\end{algorithmic}
		\end{algorithm}

		\begin{remark}\label{divergence_example}
			Selecting Bregman functions based on different constraint conditions can improve computational efficiency. For example, for the optimization problem with probabilistic simplex constraint, i.e. $\min_{x\in \Omega} f(x),$ with $\Omega=\{x=[x_1, \dots, x_d]^T \in \mathbb{R}^d | x_i \geq0, \sum_{i=1}^d x_i=1 \}$, if choosing the squared Euclidean distance, then in each projection step, projecting the iterates back onto the simplex constraint is computationally complex. However, if choosing the Bregman function $\varphi(x)= \sum_{i=1}^d x_i \log x_i$, then one can deduce that $[x_{t+1}]_j = \frac{\exp{([y_{t+1}]_j-1)}}{\sum_{i=1}^d \exp{([y_{t+1}]_i-1)}}$ and $[y_{t+1}]_j = [y_{t}]_j - \alpha \frac{\partial f(x_t)}{\partial [x_t]_j}$. The iteration points automatically remain within the set $\Omega$ without additional projection computations and the computational cost is significantly reduced.
		\end{remark}
		
		\section{Algorithm Analyses}
		In this section, we provide the analyses of the differential privacy and the regret of the proposed algorithm. Denote
		$$X_t=((x_t^1)^T,\dots,(x_t^N)^T)^T, \, {X_t}'=(({x_t^1}')^T,\dots,({x_t^N}')^T)^T,$$
		$$Z_t=((z_t^1)^T,\dots,(z_t^N)^T)^T, \, {Z_t}'=(({z_t^1}')^T,\dots,({z_t^N}')^T)^T,$$
		$$Q_t=((q_t^1)^T,\dots,(q_t^N)^T)^T, \, {Q_t}'=(({q_t^1}')^T,\dots,({q_t^N}')^T)^T,$$
		$$\Xi_t=((\xi_t^1)^T,\dots,(\xi_t^N)^T)^T, \, \overline{x}_{t}=\frac{1}{N} (\textbf{1}_N^T \otimes I_d ) X_{t}.$$
		
			\subsection{Differential privacy analysis}
		
		We firstly introduce the basic concepts related to the differential privacy.
		
		\begin{definition}\label{sensitive}
			(\cite{privacy_definition}) Given two data sets $\mathcal{F}=(f_1, f_2, \dots, f_N)$ and $\mathcal{F}'=(f_1', f_2', \dots, f_N')$, if there exists an $i \in \{1, \dots, N\}$ such that $f_i \neq f_i'$ and $f_j = f_j'$ for all $j \neq i$, then the data sets $\mathcal{F}$ and $\mathcal{F}'$ are adjacent. Two adjacent data sets is denoted as ${\rm adj}(\mathcal{F}, \mathcal{F}')$.
		\end{definition}
		
		For the distributed online optimization problem $P$, let the execution of the algorithm be $\varrho=\{X_1, Q_1, Z_1\}, \dots, \{X_T, Q_T, Z_T\}$. The observation sequence in the above execution of the algorithm is $Q_1, \dots, Q_T$ (the information exchanged during the algorithm's execution). Denote the mapping from execution to observation by $A(\varrho) \triangleq Q_1, \dots, Q_T$. Let $Obs$ be the observation sequence of the algorithm for any distributed online optimization problem $P$. For any initial state $X_1$, any fixed observation $O \in Obs$, and any two adjacent data sets $\mathcal{F}$ and $\mathcal{F}'$, we denote the algorithm's executions for the two distributed online optimization problems $P$ and $P'$ by $A^{-1}(\mathcal{F}, O, X_1)=\{X_1, Q_1, Z_1\}, \dots, \{X_T, Q_T, Z_T\}$ and $A^{-1}(\mathcal{F}', O, X_1)=\{X_1', Q_1', Z_1'\}, \dots$, $\{X_T', Q_T', Z_T'\}$, respectively. Let $A^{-1}_{X_t}(\mathcal{F}, O, X_1)$ be the set of elements $X_t$ in the execution $A^{-1}(\mathcal{F}, O, X_1)$. Using the above notation, we give the definition of differential privacy.
		
		\begin{definition}
			(\cite{DP-definition}) For any two adjacent data sets $\mathcal{F}$ and $\mathcal{F}'$, any initial state $X_1$, and any set of observation sequences $\mathcal{O} \subseteq Obs$, if the algorithm satisfies
			$$\mathbb{P}[A^{-1}(\mathcal{F}, \mathcal{O}, X_1)] \leq \exp(\epsilon) \mathbb{P}[A^{-1}(\mathcal{F}', \mathcal{O}, X_1)], $$
			where $\epsilon>0$ is the privacy parameter, then the algorithm satisfy $\epsilon$-differential privacy.
		\end{definition}
		
		If two data sets are adjacent, then the probability that the algorithm produces the same output for both data sets is very close. Differential privacy aims to mitigate the difference in outputs between two adjacent data sets by adding noises to the algorithm. To determine the amount of noises to add, the sensitivity plays an important role in algorithm design. We will now provide the definition of algorithm sensitivity.

		\begin{definition}\label{defeinition_sensitive}
			(\cite{DP-definition}\cite{DP-definition2}) At each time $t$, for any initial state $X_1$ and any two adjacent data sets $\mathcal{F}$ and $\mathcal{F}'$, the sensitivity of the algorithm is defined as
			$$\bigtriangleup (t) \triangleq \sup_{O\in Obs} \left\{ \sup_{X\in A^{-1}_{X_t}(\mathcal{F}, O, X_1),\, X'\in A^{-1}_{X_t}(\mathcal{F}', O, X_1)} \|X-X'\|_1 \right\}.$$
		\end{definition}
		
		In differential privacy, sensitivity is a crucial quantity that determines the amount of noises to be added in each iteration to achieve differential privacy. The sensitivity of an algorithm describes the extent to which a change in a single data point in adjacent data sets affects the algorithm. Therefore, we determine the noises magnitude by constraining sensitivity to ensure $\epsilon$-differential privacy.
		The following lemma provides a bound for the sensitivity of Algorithm 1.
		
		\begin{lemma}\label{sensitive_lemma}
			Suppose Assumption \ref{AS2} and Assumption \ref{AS4} hold. Then the sensitivity of Algorithm 1 satisfies
			\begin{equation}\label{sensitive_lemma_equation}
				\bigtriangleup (t) \leq \frac{2\sqrt{d}\alpha_t \theta}{\omega}.
			\end{equation}
		\end{lemma}
		
		\begin{proof}
			Based on (\ref{algorithm2}) and the first-order optimality condition, we have
			\begin{equation}\label{sensitive_lemma_equation2}
				\left\langle \alpha_t \nabla f_t^i(x_t^i) + \nabla_x D_{\varphi}(x_{t+1}^i, z_t^i), x-x_{t+1}^i \right\rangle \geq 0.
			\end{equation}
			By (\ref{bregman}), we have
			\begin{equation}\label{1010}
				\nabla_x D_{\varphi}(x_{t+1}^i, z_t^i)= \nabla \varphi(x_{t+1}^i)- \nabla \varphi(z_t^i).
			\end{equation}
			Let $x=z_t^i$, by (\ref{sensitive_lemma_equation2}) and (\ref{1010}), we obtain
			\begin{align*}
				\left\langle \alpha_t \nabla f_t^i(x_t^i), z_t^i-x_{t+1}^i \right\rangle &\geq \left\langle  \nabla \varphi(z_t^i)-\nabla \varphi(x_{t+1}^i), z_t^i-x_{t+1}^i \right\rangle\\
				&\geq \omega \|z_t^i- x_{t+1}^i\|^2.
			\end{align*}
			By Assumption \ref{AS2} and Cauchy-Schwarz inequality, we have
			$$\omega \|z_t^i- x_{t+1}^i\|^2 \leq \alpha_t\|\nabla f_t^i(x_t^i)\|\|z_t^i-x_{t+1}^i\| \leq \alpha_t \theta \|z_t^i-x_{t+1}^i\|.$$
			Therefore, we have
			\begin{equation}\label{sensitive_lemma_equation3}
				\|x_{t+1}^i-z_t^i\| \leq \frac{\alpha_t \theta}{\omega}.
			\end{equation}
			Similarly, we have
			\begin{equation}\label{sensitive_lemma_equation4}
				\|{x_{t+1}^i}'-{z_t^i}'\| \leq \frac{\alpha_t \theta}{\omega}.
			\end{equation}
			At each time $t$, by Definition \ref{defeinition_sensitive}, the attacker observes the same information from two adjacent data sets, i.e., $q_t^i = {q_t^i}'$.
			From (\ref{algorithm1}), we know that $z_t^i={z_t^i}', \forall \, t=1, \dots, T, \forall \, i =1,\dots, N$. By (\ref{sensitive_lemma_equation3}), (\ref{sensitive_lemma_equation4}) and the triangle inequality, we have
			\begin{align}
				\|A^{-1}(\mathcal{F}, O, X_1)-A^{-1}(\mathcal{F}', O, X_1)\|_1 &=
				\|x_{t+1}^i-{x_{t+1}^i}'\|_1 \nonumber\\
				&\leq \sqrt{d} \|x_{t+1}^i-{x_{t+1}^i}'\| \nonumber\\
				& \leq \|x_{t+1}^i-z_t^i\|+\|{x_{t+1}^i}'-{z_t^i}'\| \nonumber \\
				&\leq \frac{2\sqrt{d}\alpha_t \theta}{\omega}. \label{sensitive_lemma_equation5}
			\end{align}
			From the arbitrariness of observation $O$ and adjacent data sets $\mathcal{F}$ and $\mathcal{F}'$, as well as (\ref{sensitive_lemma_equation5}), we obtain (\ref{sensitive_lemma_equation}). \qed
			
		\end{proof}
		
		The data set $\mathcal{F}=\bigcup_{t=1}^T \mathcal{F}_t$ contains all the information that needs to be protected, where $\mathcal{F}_t = \{f_t^1, \dots, f_t^N\}$ represents the information that the algorithm needs to protect at time $t$. The adjacent data set of $\mathcal{F}$ is denoted as $\mathcal{F}'=\bigcup_{t=1}^T \mathcal{F}_t'$, where $\mathcal{F}_t'= \{{f_t^1}', \dots, {f_t^N}'\}$. Next, we present the $\epsilon$-differential privacy theorem.
		
		\begin{theorem}\label{theorem1}
			Suppose Assumption \ref{AS2} and Assumption \ref{AS4} hold. If $\sigma_t =\frac{\bigtriangleup (t)}{\epsilon}$, $\forall \, t=1,\dots, T$ and $\epsilon>0$, then Algorithm 1 guarantees $\epsilon$-differential privacy at each time $t$. Furthermore, over the time horizon $T$, Algorithm 1 guarantees $\hat{\epsilon}$-differential privacy, where $\hat{\epsilon}= \sum_{t=1}^T \frac{\bigtriangleup (t)}{\sigma_t}$.
		\end{theorem}
		
		\begin{proof}
			By Definition \ref{defeinition_sensitive}, we have
			$$\|X_t-{X_t}'\|_1 \leq \bigtriangleup (t).$$
			Thus, we have
			$$\sum_{i=1}^N \sum_{j=1}^d \left| [x_t^i]_j- [{x_t^i}]_j \right| = \|X_t-{X_t}'\|_1 \leq \bigtriangleup (t),$$
			where $[x_t^i]_j$ denotes the $j$-th component of the vector $x_t^i$. By the properties of the Laplace distribution and $q_t^i = {q_t^i}'$, it follows that
			\begin{align}\label{epsilon_equation}
				 \prod_{i=1}^N \prod_{j=1}^d \frac{\mathbb{P}\left[[q_t^i]_j-[x_t^i]_j\right]}{\mathbb{P}\left[[{q_t^i}']_j-[{x_t^i}']_j\right]}
				&=\prod_{i=1}^N \prod_{j=1}^d \frac{\exp \left(-\frac{\left|[q_t^i]_j-[x_t^i]_j\right|}{\sigma_t}\right)}{\exp \left(-\frac{\left|[{q_t^i}']_j-[{x_t^i}']_j\right|}{\sigma_t}\right)}  \nonumber\\
				& \leq \prod_{i=1}^N \prod_{j=1}^d \exp \left(\frac{\left|[{q_t^i}']_j-[{x_t^i}']_j-[q_t^i]_j+[x_t^i]_j\right|}{\sigma_t}\right)  \nonumber\\
				& =\exp \left(\sum_{i=1}^N \sum_{j=1}^d \frac{\left|[x_t^i]_j-[{x_t^i}']_j\right|}{\sigma_t}\right) \nonumber \\
				&  \leq \exp \left(\frac{\triangle(t)}{\sigma_t}\right).
			\end{align}
			By (\ref{epsilon_equation}), we obtain
			\begin{align*}
				\mathbb{P}[A^{-1}(\mathcal{F}_t, \mathcal{O}, X_1)] \leq \exp(\epsilon) \mathbb{P}[A^{-1}(\mathcal{F}_t', \mathcal{O}, X_1)].
			\end{align*}
			Thus, the Algorithm 1 guarantees $\epsilon$-differential privacy at each time $t$. Let $\hat{\epsilon} = \sum_{t=1}^T \frac{\bigtriangleup (t)}{\sigma_t}$, by (\ref{epsilon_equation}), we
			obtain
			\begin{align*}
				\mathbb{P}[A^{-1}(\mathcal{F}, \mathcal{O}, X_1)] &= \prod_{t=1}^T \mathbb{P}[A^{-1}(\mathcal{F}_t, \mathcal{O}, X_1)]\\
				&\leq \prod_{t=1}^T \exp \left(\frac{\triangle(t)}{\sigma_t}\right)\mathbb{P}[A^{-1}(\mathcal{F}_t', \mathcal{O}, X_1)]\\
				& = \exp(\hat{\epsilon})\mathbb{P}[A^{-1}(\mathcal{F}', \mathcal{O}, X_1)].
			\end{align*}
			\qed
		\end{proof}
		
		\begin{remark}
			According to Theorem \ref{theorem1}, Algorithm 1 satisfies $\epsilon$-differential privacy in each iteration. Since the upper bound for the sensitivity includes the step size, if choosing a decaying step size, then the sensitivity will gradually decrease. The level of differential privacy depends on the step size $\alpha_t$, the bound on the gradient $\theta$, the dimension of the decision variables $d$, the strong convexity parameter $\omega$, and the noise magnitude $\sigma_t$. The smaller $\epsilon$ is, the higher the level of differential privacy.
		\end{remark}
		
		\subsection{Regret analysis}
		
		This subsection analyzes the regret of Algorithm 1. By appropriately choosing parameters, for Algorithm 1, this section provides an upper bound $O(\sqrt{T})$ on its regret. We present some lemmas for analyzing the regret of the algorithm.
		
		For any $t \geq s \geq 1$, the state transition matrix is defined as
		$$
		\Phi(t, s)=\left\{\begin{array}{l}
			A(t-1) \cdots A(s+1) A(s), \text { if }\,\, t>s, \\
			I_N, \text { if }\,\, t=s.
		\end{array}\right.
		$$
		
		According to Property 1 in \cite{stochastic_matrix}, we have the following lemma.
		
		\begin{lemma}\label{graph_lemma}
			Suppose Assumption \ref{AS1} hold. For any $i, j \in V$, $t \geq s$, we have
			\begin{equation}
				\qquad\left|[\Phi(t, s)]_{i j}-\frac{1}{N}\right| \leq C \lambda^{t-s},
			\end{equation}
			where $C=2\left(1+a^{-(N-1)B}\right) /\left(1+a^{(N-1)B}\right)$ and $\lambda=  \left(1-a^{(N-1)B}\right)^{1 /(N-1)B}$.
		\end{lemma}
		
		In order to prove the main results, some necessary lemmas are provided.
		\begin{lemma}\label{x_i_difference}
			Suppose Assumptions \ref{AS1}-\ref{AS3} hold. For the sequence of decisions $\{x_t^i, \, t=1,\dots,T, \, i=1,\dots,N\}$ generated  by Algorithm 1, for any $i,j \in \{1,\dots, N\}$, we have
			\begin{equation}\label{x_i_difference_equation}
				\begin{aligned}
				&\quad \mathbb{E}[\|x_t^i-x_t^j\|]\\
				& \leq \! 2\sqrt{Nd}C\lambda^{t-1}\mathbb{E}\left[\|X_1\|\right]+2\sqrt{2}NdC\sum_{k=1}^{t-1}\lambda^{t-k}\sigma_k + \frac{2\sqrt{d}N\theta C}{\omega}\sum_{k=1}^{t-1}\lambda^{t-k-1}\alpha_k.
				\end{aligned}
			\end{equation}
		\end{lemma}
		
		\begin{proof}
			Firstly, denote $p_t^i= x_{t+1}^i - z_t^i$ and $P_t= ((p_t^1)^T,\dots,(p_t^N)^T)^T$. By (\ref{sensitive_lemma_equation3}), we have
			\begin{equation}\label{x_i_difference_equation1}
				\|p_t^i\|\leq \frac{\alpha_t \theta}{\omega}.
			\end{equation}
			By (\ref{algorithm1}), we obtain
			\begin{align}
				X_t &=Z_{t-1}+P_{t-1} \nonumber \\
				&= \left(A(t-1) \otimes I_d\right)\left(X_{t-1}+ \Xi_{t-1}\right)+P_{t-1} \nonumber\\
				&= \left( \Phi(t,1) \otimes I_d\right)X_1 + \sum_{k=1}^{t-1}\left( \Phi(t,k) \otimes I_d\right)\Xi_k+ \sum_{k=1}^{t-1}\left( \Phi(t,k+1) \otimes I_d\right)P_k. \label{x_i_difference1}
			\end{align}
			By (\ref{x_i_difference1}), we obtain
			\begin{align}
				x_t^i= \left( [\Phi(t,1)]_i \otimes I_d\right)X_1
				\!+\! \sum_{k=1}^{t-1}\left( [\Phi(t,k)]_i \otimes I_d\right)\Xi_k \!+\! \sum_{k=1}^{t-1}\left( [\Phi(t,k+1)]_i \otimes I_d\right)P_k. \label{x_i_difference2}
			\end{align}
			By Assumption \ref{AS1} and (\ref{x_i_difference1}), we have
			\begin{align}\label{x_average}
				\overline{x}_{t} &=\frac{1}{N}(\textbf{1}_N^T \!\otimes I_d ) X_{t} \nonumber \\
				&=\frac{1}{N}(\textbf{1}_N^T \otimes I_d ) \left( \Phi(t,1) \otimes I_d\right)X_1+ \frac{1}{N}\sum_{k=1}^{t-1}(\textbf{1}_N^T \otimes I_d )\left( \Phi(t,k) \otimes I_d\right)\Xi_k \nonumber \\
				&\quad + \frac{1}{N}\sum_{k=1}^{t-1}(\textbf{1}_N^T \otimes I_d )\left( \Phi(t,k+1) \otimes I_d\right)P_k \nonumber \\
				&=\overline{x}_{1}+ \frac{1}{N}\sum_{k=1}^{t-1}\left(\textbf{1}_N^T \otimes  I_d\right)\Xi_k + \frac{1}{N}\sum_{k=1}^{t-1}\left(\textbf{1}_N^T \otimes  I_d\right)P_k.
			\end{align}
			Combining Lemma \ref{graph_lemma}, (\ref{x_i_difference_equation1}), (\ref{x_i_difference1}), and (\ref{x_i_difference2}), we obtain
			\begin{align}\label{x_i_difference_equation2}
				\mathbb{E}[\|x_t^i-\overline{x}_{t}\|]&= \mathbb{E}\left[\left\|\left(\left([\Phi(t,1)]_i-\frac{1}{N}\textbf{1}_N^T\right)\otimes I_d\right)X_1\right\|\right] \nonumber \\
				&\quad +\mathbb{E}\left[\left\|\sum_{k=1}^{t-1}\left(\left([\Phi(t,k)]_i-\frac{1}{N}\textbf{1}_N^T\right)\otimes I_d\right)\Xi_k\right\|\right] \nonumber \\
				&\quad +\mathbb{E}\left[\left\|\sum_{k=1}^{t-1}\left(\left([\Phi(t,k+1)]_i-\frac{1}{N}\textbf{1}_N^T\right)\otimes I_d\right)P_k\right\|\right] \nonumber \\
				&\leq \left\|\left(\left([\Phi(t,1)]_i-\frac{1}{N}\textbf{1}_N^T\right)\otimes I_d\right)\right\|\mathbb{E}\left[\|X_1\|\right] \nonumber \\
				&\quad +\sum_{k=1}^{t-1}\left\|\left(\left([\Phi(t,k)]_i-\frac{1}{N}\textbf{1}_N^T\right)\otimes I_d\right)\right\|\mathbb{E}\left[\|\Xi_k\|\right] \nonumber \\
				& \quad +\sum_{k=1}^{t-1}\left\|\left(\left([\Phi(t,k+1)]_i-\frac{1}{N}\textbf{1}_N^T\right)\otimes I_d\right)\right\|\mathbb{E}\left[\|P_k\|\right] \nonumber \\
				&\leq \!\!\sqrt{Nd}C\lambda^{t-1}\mathbb{E}\left[\| X_1 \!\|\right]\! \!+\!\sqrt{2}NdC \!\!\sum_{k=1}^{t-1}\!\lambda^{t-k}\sigma_k \!+\! \frac{\sqrt{d}N\theta C}{\omega}\!\!\sum_{k=1}^{t-1}\!\lambda^{t-k-1}\alpha_k,
			\end{align}
			where the last inequality holds due to Lemma \ref{graph_lemma}, (\ref{x_i_difference_equation1}), and $\mathbb{E}\left[\|\Xi_k\|\right] = \sqrt{2Nd}\sigma_t$. Therefore, for any $i, j \in \{1, \dots, N\}$, by (\ref{x_i_difference_equation2}) and the triangle inequality, we have
			(\ref{x_i_difference_equation}). \qed
		\end{proof}
		
		\begin{lemma}\label{extra_lemma}
			Suppose Assumptions \ref{AS1}-\ref{AS4} hold. For the sequence of decisions $\{x_t^i, \, t=1,\dots,T, \, i=1,\dots,N\}$ generated by Algorithm 1, for any $i\in \{1,\dots, N\}$, $x\in \Omega$, we have
			\begin{equation}\label{extra_lemma_1}
				\sum_{t=1}^T \sum_{i=1}^N \frac{1}{\alpha_t} \mathbb{E} \left[ D_{\varphi}(x, z_t^i)-D_{\varphi}(x, x_{t+1}^i)  \right] \leq  \frac{NM\eta^2}{2\alpha_T} +\sum_{t=1}^T\sum_{i=1}^N  \frac{M}{2\alpha_t} \mathbb{E}\left[\|\xi_t^i\|^2\right].
			\end{equation}
		\end{lemma}
		
		\begin{proof}
			By Assumption \ref{AS1}, Assumption \ref{AS4} and (\ref{algorithm1}), for any $x \in \Omega$, we have
			\begin{align}
				\sum_{i=1}^N D_{\varphi}(x, z_t^i) &= \sum_{i=1}^ND_{\varphi}\left(x, \sum_{j=1}^N a_{ij}(t)(x_t^j+\xi_t^j)\right)  \nonumber \\
				&\leq \sum_{i=1}^N\sum_{j=1}^N a_{ij}(t)D_{\varphi}\left(x, x_t^j+\xi_t^j\right)  \nonumber \\
				&= \sum_{j=1}^N D_{\varphi}\left(x, x_t^j+\xi_t^j\right).  \nonumber
			\end{align}
			From (\ref{algorithm1}) and (\ref{algorithm2}), it is known that $x_t^i$ and $\xi_t^i$ are independent. Combining the properties of the Laplace distribution and Assumption \ref{AS4}, we obtain
			\begin{align}
				\mathbb{E}\left[D_{\varphi}\left(x, x_t^i+\xi_t^i\right)- D_{\varphi}\left(x, x_t^i\right)\right]
				& \leq \mathbb{E}\left[\langle \nabla_y D_{\varphi}\left(x, x_t^i\right), \xi_t^i \rangle \right] + \frac{M}{2}\mathbb{E}\left[\|\xi_t^i\|^2 \right]  \nonumber \\
				&= \frac{M}{2} \mathbb{E} \left[\|\xi_t^i\|^2\right].  \label{extra_lemma_equation1}
			\end{align}
			By Assumptions \ref{AS2} - \ref{AS4}, we obtain
			\begin{align}\label{th2_proof3}
				&\quad \sum_{t=1}^T \frac{1}{\alpha_t} \left(D_{\varphi}(x, x_t^i)-D_{\varphi}(x, x_{t+1}^i)  \right)  \nonumber \\
				&= \frac{1}{\alpha_1} D_{\varphi}(x, x_1^i)- \frac{1}{\alpha_T}D_{\varphi}(x, x_{T+1}^i) +\sum_{t=2}^{T-1}\left( \frac{1}{\alpha_t}- \frac{1}{\alpha_{t-1}}\right)D_{\varphi}(x, x_{t}^i)  \nonumber  \\
				& \leq \frac{M}{2\alpha_1}\|x-x_1^i\|^2 + \frac{M}{2}\sum_{t=2}^{T-1}\left( \frac{1}{\alpha_t}- \frac{1}{\alpha_{t-1}}\right)\|x-x_t^i\|^2  \nonumber  \\
				& \leq \frac{M\eta^2}{2\alpha_1}+ \frac{M\eta^2}{2}\sum_{t=2}^{T-1}\left( \frac{1}{\alpha_t}- \frac{1}{\alpha_{t-1}}\right)   \nonumber \\
				&= \frac{M\eta^2}{2\alpha_T},
			\end{align}
			where the first inequality holds due to $D_{\varphi}(x, x_{T+1}^i)\geq 0$ and
			$$D_{\varphi}(x, x_1^i)=D_{\varphi}(x, x_1^i)-D_{\varphi}(x, x) \leq \frac{M}{2}\|x-x_1^i\|,$$
			and the second inequality holds due to $\forall x,y\in \! \Omega$, $\|x  -y\| \! \leq \! \eta$. By (\ref{extra_lemma_equation1}) and (\ref{th2_proof3}), we have (\ref{extra_lemma_1}). \qed
		\end{proof}
		
		\begin{lemma}\label{bregman_lemma}
			(\cite{mirror_descent}) Consider the following problem
			$$\min_{x\in \Omega}  \{D_{\varphi}(x, y)+ \langle s, x\rangle\},$$
			where $s: \mathbb{R}^d \rightarrow \mathbb{R}$ is a convex function, and $D_{\varphi}(x, y)$ is given by (\ref{bregman}). $x^*$ is the optimal solution to the above problem if and only if
			$$\langle s, x^*-z \rangle \leq D_{\varphi}(z, y)-D_{\varphi}(z, x^*)-D_{\varphi}(x^*, y),  \, \forall \, z \in \Omega.$$
		\end{lemma}

		\begin{lemma}\label{keylemma}
			Suppose Assumptions \ref{AS1}-\ref{AS4} hold. For the sequence of decisions $\{x_t^i, \, t=1,\dots,T, \, i=1,\dots,N\}$ generated by Algorithm 1, for any $x\in \Omega$, we have
			\begin{align}\label{keylemma_equation1}
				&\quad \mathbb{E} \left[ \sum_{t=1}^T\sum_{j=1}^N \langle \nabla f_t^j(x_t^i), x_t^i- x \rangle \right] \nonumber \\
				&\leq \sum_{t=1}^T\sum_{j=1}^N\beta\eta \mathbb{E}\left[\|x_t^i-x_t^j\|\right]+\sum_{t=1}^T\sum_{j=1}^N\theta \mathbb{E}\left[\|x_t^i- z_t^j\|\right]+\frac{NM\eta^2}{2\alpha_T} \nonumber \\
				&\quad + \sum_{t=1}^T\sum_{j=1}^N \mathbb{E}\left[\langle \nabla f_t^j(x_t^j), z_t^j- x_{t+1}^j \rangle\right]
				+\sum_{t=1}^T\sum_{j=1}^N \frac{M}{2\alpha_t}\mathbb{E}\left[\|\xi_t^j\|^2\right].
			\end{align}
		\end{lemma}
		
		\begin{proof}
			Denote $s=\alpha_t \nabla f_t^j(x_t^j)$, $y=z_t^j$, $x^*= x_{t+1}^j$, and $z=x\in \Omega$. By Lemma \ref{bregman_lemma}, we have
			$$D_{\varphi}(x, x_{t+1}^j)+D_{\varphi}(x_{t+1}^j, z_t^j)-D_{\varphi}(x, z_t^j) \leq \langle \alpha_t \nabla f_t^j(x_t^j), x- x_{t+1}^j\rangle.$$
			From the non-negativity of Bregman divergence and the above inequality,
			we have
			\begin{align}
				D_{\varphi}(x, x_{t+1}^j)-D_{\varphi}(x, z_t^j) &\leq \langle  \alpha_t \nabla f_t^j(x_t^j), x- x_{t+1}^j\rangle \nonumber \\
				&= \langle \alpha_t \nabla f_t^j(x_t^j), x- z_{t}^j\rangle + \langle \alpha_t \nabla f_t^j(x_t^j), z_t^j- x_{t+1}^j\rangle. \nonumber
			\end{align}
			Rearranging and summing the above inequality yields
			\begin{align}\label{keylemma_equation2}
				&\quad \sum_{t=1}^T\sum_{j=1}^N \mathbb{E}\left[ \langle  \nabla f_t^j(x_t^j), z_{t}^j-x\rangle \right] \nonumber \\
				& \leq \sum_{t=1}^T\sum_{j=1}^N \mathbb{E}\left[ \langle  \nabla f_t^j(x_t^j), z_{t}^j-x_{t+1}^j\rangle \right] +\sum_{t=1}^T \sum_{j=1}^N \frac{1}{\alpha_t} \mathbb{E} \left[ D_{\varphi}(x, z_t^j)-D_{\varphi}(x, x_{t+1}^j)  \right].
			\end{align}
			By Assumptions \ref{AS1}- \ref{AS3}, for any $i, j \in \{1,\dots, N\}$, $ x\in \Omega$, we have
			\begin{align}\label{keylemma_equation3}
				&\quad \mathbb{E}\left[\sum_{t=1}^T \sum_{j=1}^N \langle \nabla f_t^j(x_t^i), x_{t}^i-x\rangle \right] \nonumber \\
				&=\! \sum_{t=1}^T \!\sum_{j=1}^N \! \mathbb{E}\!\left[\!\langle \nabla f_t^j(x_t^i) \!-\!\! \nabla f_t^j(x_t^j), x_t^i \!-\! x \rangle \!+\! \langle \nabla f_t^j(x_t^j), z_t^j \!-\! x \rangle \!+\! \langle \nabla f_t^j(x_t^j), x_t^i \!-\! z_t^j \rangle \!\right] \nonumber \\
				& \leq \sum_{t=1}^T \sum_{j=1}^N \mathbb{E} \left[\beta\eta\|x_t^i - x_t^j\| + \theta\| x_t^i - z_t^j\| + \langle \nabla f_t^j(x_t^j), z_t^j-x \rangle \right],
			\end{align}
			where the inequality holds due to $\|\nabla f_t^j(x_t^i) - \nabla f_t^j(x_t^j)\| \leq \beta\|x_t^i -x_t^j\|$, $\|x_t^i -x\| \leq \eta$, and $\|\nabla f_t^j(x_t^j)\| \leq \theta$.
			Combining (\ref{extra_lemma_1}), (\ref{keylemma_equation2}), and (\ref{keylemma_equation3}), the lemma is thus proved. \qed

		\end{proof}
		
		Next, we present the theorem for the regret analysis of Algorithm 1.
		
		\begin{theorem}\label{theorem2}
			Suppose Assumptions \ref{AS1}-\ref{AS4} hold. If $\alpha_t=\frac{1}{N\sqrt{t}}$, $\sigma_t =\frac{2\sqrt{d}\alpha_t\theta}{\omega \epsilon}$, $t=1,\dots, T$, for any $i \in \{1,\dots, N\}$, then the regret of Algorithm 1 satisfies
			\begin{equation}\label{theorem2-equation1}
				\begin{aligned}
					\mathbb{E}[R_T^i] \leq U_1+U_2\sqrt{T},
				\end{aligned}
			\end{equation}
			where $U_1=\frac{2 N^{\frac{3}{2}}\sqrt{d}C(\beta\eta+2\theta)\mathbb{E}\left[\|X_1\|\right]}{1-\lambda}$ and $U_2=\frac{8\sqrt{2N}d\theta^2}{\omega\epsilon}+ \frac{2\sqrt{N}\theta^2}{\omega}+\frac{M\eta^2}{2}+ \frac{8d^2M\theta^2}{\omega^2 \epsilon^2}+\frac{8\sqrt{2}d^{\frac{3}{2}}NC\theta(\beta\eta+2\theta)}{\omega\epsilon(1-\lambda)} +\frac{4\sqrt{d}NC\theta(\beta\eta+2\theta)}{\omega(1-\lambda)}$, which implies $\lim_{T \rightarrow \infty} \mathbb{E}[R_T^i]/T =0$.
		\end{theorem}
		
		\begin{proof}
			To prove $\max_{x\in {\Omega}} \left(  \mathbb{E}\left[ \sum_{t=1}^T \sum_{j=1}^N \langle \nabla f_t^j(x_t^i), x_t^i-x\rangle  \right] \right)$ is sublinear, it suffices to show $\mathbb{E}\left[  \sum_{t=1}^T \sum_{j=1}^N \langle \nabla f_t^j(x_t^i), x_t^i-x\rangle \right]$ is sublinear for any $x \in \Omega$. By (\ref{algorithm1}), $\|\nabla f_t^j(x_t^j)\| \leq \theta$, and $\mathbb{E}[\xi_t^i]=0$, we have
			\begin{align}\label{theorem2-equation3}
				\mathbb{E}\left[\langle \nabla f_t^j(x_t^j), z_t^j- x_{t+1}^j\rangle \right] &=\mathbb{E}\left[\langle \nabla f_t^j(x_t^j), \sum_{i=1}^N a_{ji}(t)(x_t^i + \xi_t^i)- x_{t+1}^j\rangle \right] \nonumber \\
				&=\!\! \sum_{i=1}^N \! a_{ji}(t) \!\!\left(\mathbb{E}[\langle \nabla f_t^j(x_t^j), x_t^i \!-\! x_{t+1}^j \rangle] \!+\! \mathbb{E}[\!\langle \nabla f_t^j(x_t^j), \xi_t^i \rangle] \right)\nonumber \\
				&\leq \sum_{i=1}^N a_{ji}(t)\theta \mathbb{E}[\|x_t^i- x_{t+1}^j\|].
			\end{align}
			By (\ref{algorithm1}), the triangle inequality, and $\mathbb{E}[\|\xi_t^h\|]=\sqrt{2d}\sigma_t$, we have
			\begin{align}\label{theorem2-equation4}
				\mathbb{E}[\|x_t^i - z_t^j\|] &\leq \sum_{h=1}^N a_{jh}(t)\mathbb{E}[\|x_t^i - x_t^h\|] + \mathbb{E}\left[\sum_{h=1}^N a_{jh}(t)\|\xi_t^h\|\right] \nonumber\\
				& \leq \sum_{h=1}^N a_{jh}(t)\mathbb{E}[\|x_t^i - x_t^h\|] + \sqrt{2d}\sigma_t.
			\end{align}
			By (\ref{x_i_difference_equation}), (\ref{x_average}), and the properties of the Laplace distribution, we obtain
			\begin{align}
				\mathbb{E}\left[\|\overline{x}_{t+1}-\overline{x}_{t}\| \right]
				&= \mathbb{E}\left[\left\|\frac{1}{N}\left(\textbf{1}_N^T \otimes  I_d\right)\Xi_t + \frac{1}{N}\left(\textbf{1}_N^T \otimes  I_d\right)P_t\right\|\right] \nonumber\\
				&\leq \sqrt{2dN}\sigma_t + \frac{\sqrt{N}\theta \alpha_t}{\omega}. \label{theorem2-equation41}
			\end{align}
			By the triangle inequality, (\ref{x_i_difference_equation2}), and (\ref{theorem2-equation41}), we obtain
			\begin{align}\label{theorem2-equation44}
				\mathbb{E}\left[\|x_{t+1}^j - x_t^i\| \right]&\leq \mathbb{E}\left[\|x_{t+1}^j-\overline{x}_{t+1}\| \right]+\mathbb{E}\left[\|x_t^i-\overline{x}_{t}\| \right]+ \mathbb{E}\left[\|\overline{x}_{t+1}-\overline{x}_{t}\| \right] \nonumber \\
				& \leq 2\sqrt{Nd}C\lambda^{t}\mathbb{E}\left[\|X_1\|\right]+2\sqrt{2}NdC\sum_{k=1}^{t-1}\lambda^{t-k}\sigma_k +\sqrt{2dN}\sigma_t \nonumber \\
				&\quad  + \frac{2\sqrt{d}N\theta C}{\omega}\sum_{k=1}^{t-1}\lambda^{t-k-1}\alpha_k+ \frac{\sqrt{N}\theta \alpha_t}{\omega}.
			\end{align}
			Substituting (\ref{theorem2-equation3}), (\ref{theorem2-equation4}), and (\ref{theorem2-equation44}) into (\ref{keylemma_equation1}), and combining with $\boldsymbol{1}^T A(t) = \boldsymbol{1}^T$, $A^T(t) \boldsymbol{1} = \boldsymbol{1}$ and (\ref{x_i_difference_equation}), we obtain
			\begin{align*}
				&\quad \mathbb{E}\left[ \sum_{t=1}^T\sum_{j=1}^N \langle \nabla f_t^j(x_t^i), x_t^i- x \rangle\right]\\
				&\leq \sum_{t=1}^T\sum_{j=1}^N\beta\eta \mathbb{E}[\|x_t^i-x_t^j\|]+\sum_{t=1}^T\sum_{j=1}^N\sum_{h=1}^N a_{jh}(t) \theta\mathbb{E}[\|x_t^i - x_t^h\|] + \sum_{t=1}^T N\theta\sqrt{2d}\sigma_t\\
				&\quad +\!\! \sum_{t=1}^T\sum_{j=1}^N \sum_{i=1}^N \! a_{ji}(t)\theta \mathbb{E}[\|x_t^i- x_{t+1}^j\|] \!+\!\frac{NM\eta^2}{2\alpha_T} \!+ \!+\! \sum_{t=1}^T\sum_{j=1}^N \frac{M}{2\alpha_t}\mathbb{E} \left[\|\xi_t^j\|^2\right]\\
				& \leq 2CN\sqrt{d}(\beta \eta \!+\! \theta)\!\! \sum_{t=1}^T \!\! \left( \!\!\!\sqrt{N}\lambda^{t-1}\mathbb{E}\!\left[\|X_1\|\right] \!+\!\sqrt{2d}N\sum_{k=1}^{t-1}\!\lambda^{t-k}\sigma_k \!+\! \frac{N\theta }{\omega}\!\sum_{k=1}^{t-1}\!\lambda^{t-k-1}\alpha_k \!\!\right)\\
				&\quad + 2C\theta\sum_{t=1}^T \!\left(\! N^{\frac{3}{2}}\sqrt{d} \!\lambda^{t}\mathbb{E}\left[\|X_1\|\right] \!+\!\sqrt{2}N^2d \!\sum_{k=1}^{t-1} \!\!\lambda^{t-k}\sigma_k+ \frac{\sqrt{d}N^2\theta}{\omega}\sum_{k=1}^{t-1}\lambda^{t-k-1}\alpha_k \right)\\
				& \quad +\sum_{t=1}^T\left(2\sqrt{2d}N^{\frac{3}{2}}\theta\sigma_t + \frac{N^{\frac{3}{2}}\theta^2 \alpha_t}{\omega}+  \frac{NMd\sigma_t^2}{\alpha_t}\right) +\frac{NM\eta^2}{2\alpha_T}.
			\end{align*}
			Substituting $\alpha_t=\frac{1}{N\sqrt{t}}$ and $\sigma_t =\frac{2\sqrt{d}\alpha_t \theta}{\omega \epsilon}$ into the above inequality, and combining with $\sum_{t=1}^T \lambda^{t} \leq \frac{1}{1-\lambda}$, $\sum_{t=1}^T \alpha_t\leq \frac{2\sqrt{T}}{N}$, and
			$$\sum_{t=1}^T \sum_{k=1}^{t-1} \lambda^{t-k}\alpha_k \leq \sum_{t=1}^T\sum_{k=1}^t \lambda^{t-k}\alpha_k =\sum_{k=1}^T\alpha_k\sum_{t=k}^T \lambda^{t-k} \leq \frac{2\sqrt{T}}{N(1-\lambda)},$$ we obtain
			\begin{align}\label{theorem2-equation5}
				&\quad \mathbb{E}\left[ \sum_{t=1}^T\sum_{j=1}^N \langle \nabla f_t^j(x_t^i), x_t^i- x \rangle\right] \nonumber\\
				&\leq 2 N^{\frac{3}{2}}\sqrt{d}C(\beta\eta+2\theta)\mathbb{E}\left[\|X_1\|\right]\sum_{t=1}^T \lambda^{t-1} +2\sqrt{2d}N^{\frac{3}{2}}\theta \sum_{t=1}^T \sigma_t +  \frac{N^{\frac{3}{2}}\theta^2}{\omega}\sum_{t=1}^T \alpha_t\nonumber\\
				&\quad + \frac{NM\eta^2}{2\alpha_T}+ NMd\sum_{t=1}^T \frac{\sigma_t^2}{\alpha_t} + 2\sqrt{2}N^2dC(\beta\eta+2\theta)\!\sum_{t=1}^T \!\sum_{k=1}^{t-1}\!\!\lambda^{t-k}\sigma_k \nonumber\\
				&\quad +\! \frac{2\sqrt{d}N^2\theta C(\beta\eta+2\theta)}{\omega}\!\!\sum_{t=1}^T \!\sum_{k=1}^{t-1}\lambda^{t-k-1}\alpha_k \nonumber\\
				&\leq \frac{2 N^{\frac{3}{2}}\sqrt{d}C(\beta\eta+2\theta)\mathbb{E}\left[\|X_1\|\right]}{1-\lambda} + \frac{8\sqrt{2N}d\theta^2}{\omega\epsilon}\sqrt{T} + \frac{2\sqrt{N}\theta^2}{\omega}\sqrt{T}+\frac{M\eta^2}{2}\sqrt{T} \nonumber\\
				&\quad +  \frac{8d^2M\theta^2}{\omega^2 \epsilon^2}\sqrt{T}+\frac{8\sqrt{2}d^{\frac{3}{2}}NC\theta(\beta\eta+2\theta)}{\omega\epsilon(1-\lambda)}\sqrt{T} +\frac{4\sqrt{d}NC\theta(\beta\eta+2\theta)}{\omega(1-\lambda)}\sqrt{T}.
			\end{align}
			By the arbitrariness of $x$ and combining with (\ref{theorem2-equation5}), we obtain (\ref{theorem2-equation1}). \qed
		\end{proof}
		
		\begin{remark}
			From Theorem \ref{theorem2}, we know that Algorithm 1 converges sublinearly. This is consistent with the best results in existing works on distributed online nonconvex optimization. The differential privacy mechanism we added does not affect the order of the upper bound on the regret, but only the coefficient. Furthermore, the regret of the algorithm depends on the properties of the cost function, the constraint set, the problem's dimension, the number of nodes, the connectivity of the communication graph, and the privacy parameters.
			It is worth noting that from Theorem \ref{theorem2}, we can see that the larger the privacy parameter $\epsilon$, the tighter the bound on the regret. However, in the differential privacy analysis of Theorem \ref{theorem1}, the smaller the privacy parameter, the higher the privacy level. Therefore, in practical applications, it is necessary to choose an appropriate privacy parameter to balance algorithm performance and privacy level.
		\end{remark}

		\section{Numerical Simulations}
		In this section, we demonstrate the theoretical results of the proposed algorithm (DPDO-NC) through the distributed localization problem (Example \ref{nonconvexexample}). We consider a time-varying communication topology consisting of 6 sensors. The time-varying communication topology switches between three graphs, denoted by $G_1, G_2, G_3$. The weight matrices for these three graphs are denoted as $A_1, A_2, A_3$, where
		
		{\footnotesize
		\begin{equation}\nonumber
			A_1\!\!=\!\!\!\left[\begin{array}{cccccc}
				\frac{1}{2} & 0 & 0 &0  & 0 & \frac{1}{2} \\
				\frac{1}{2} & \frac{1}{2} & 0 & 0 &0 &0  \\
				0 & \frac{1}{2} & \frac{1}{2} & 0 & 0 & 0  \\
				0 & 0 & \frac{1}{2} & \frac{1}{2} & 0 & 0  \\
				0 & 0 & 0 & \frac{1}{2} & \frac{1}{2} & 0  \\
				0 & 0 & 0 & 0 & \frac{1}{2} & \frac{1}{2} \\
			\end{array}\right]\!\!\!, \, \! A_2\!\!=\!\!\!\left[\begin{array}{cccccc}
				0 & \frac{1}{5}  & \frac{1}{5}  & \frac{1}{5}  & \frac{1}{5} & \frac{1}{5} \\
				\frac{1}{5}  & 0 & \frac{1}{5}  & \frac{1}{5}  & \frac{1}{5}  & \frac{1}{5} \\
				\frac{1}{5} & \frac{1}{5} &0 &\frac{1}{5} &\frac{1}{5} &\frac{1}{5}  \\
				\frac{1}{5} & \frac{1}{5} & \frac{1}{5}&0 &\frac{1}{5} &\frac{1}{5}  \\
				\frac{1}{5} & \frac{1}{5}& \frac{1}{5} &\frac{1}{5} & 0&\frac{1}{5} \\
				\frac{1}{5} &\frac{1}{5} &\frac{1}{5} &\frac{1}{5} &\frac{1}{5} &0 \\
			\end{array}\right]\!\!\!, \,\! A_3\!\!=\!\!\!\left[\begin{array}{cccccc}
				\frac{1}{3} & 0 &\frac{1}{3}& 0 & \frac{1}{3} & 0\\
				0 & \frac{1}{3} & 0 &\frac{1}{3}& 0 & \frac{1}{3} \\
				\frac{1}{3} & 0 & \frac{1}{3} & 0 & \frac{1}{3}&0 \\
				0 & \frac{1}{3} & 0 & \frac{1}{3} & 0 &\frac{1}{3}\\
				\frac{1}{3} & 0 & \frac{1}{3} & 0 & \frac{1}{3} & 0\\
				0 &\frac{1}{3}&0 &\frac{1}{3}& 0 & \frac{1}{3}  \\
			\end{array}\right]\!\!\!.
		\end{equation}}
	
	The evolution process of the target location is defined as
	$$x_{t+1}^0 =x_t^0+\begin{bmatrix}\frac{(-1)^{q_t}\sin(t/50)}{10t}\\\frac{-q_t\cos(t/70)}{40t}\end{bmatrix},$$
	where $q_t \sim Bernoulli(0.5)$, and initial state $x_1^0 = [0.8, 0.95]^T$. The distance measurements are given by
	$$d_t^i = \|s_i -x_t^0\| +\vartheta_t^i, \,\,\, i=1,\dots, N,$$
	where $\vartheta_t^i$ is the measurement error of the model, which follows a uniform distribution over $[0,  0.001]$. The sensors collaborate to solve the following problem
	$$\min_{x\in \Omega } \sum_{i=1}^6 \frac{1}{2}\left| \|s_i -x\| - d_t^i \right|^2,$$
	where $\Omega = \{x\in \mathbb{R}^2 | \|x\|_1 \leq 3\}$.  The initial positions of the sensors are $s_i=[0.8, 0.95]^T, i=1, \dots, 6$. The initial estimates of the sensors are $x_1^i=[0, 0]^T, i=1, \dots, 6$. Fix the number of iterations to $T=500$. The distance function of the algorithm is $\varphi = \frac{1}{2}\|x\|^2$, and the step size is $\alpha_t = \frac{1}{6\sqrt{t}}$.
	
	Figure \ref{fig1} shows the evolution of $\max_i \mathbb{E}[R_T^{i}]/T$ with different privacy levels, where the privacy parameters are set to $\epsilon=5$, $1$, and $0.5$, and the cases without differential privacy is also compared. From Figure \ref{fig1}, it can be observed that Algorithm 1 converges under different privacy levels, and the lower the privacy level, i.e., the larger $\epsilon$, the better the convergence performance of the algorithm. The parameter $\epsilon$ reveals the trade-off between privacy protection and algorithm performance, which is consistent with the theoretical results presented in this paper.
	
	\begin{figure}[htbp]
		\centering
		\includegraphics[scale=0.6]{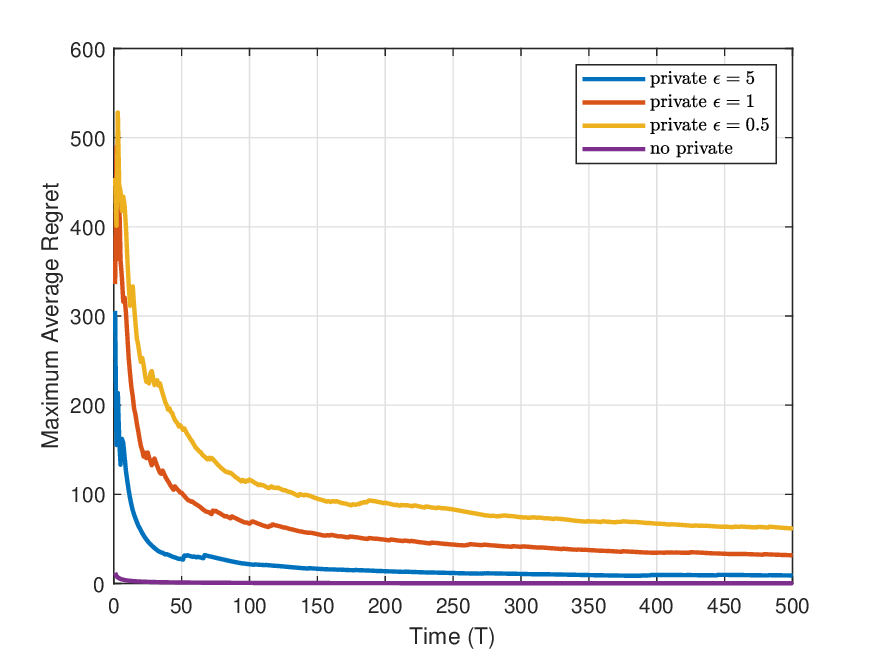}
		\caption{ Evolution of $\max_i \mathbb{E}[R_T^{i}]/T$ with different privacy levels.}
		\label{fig1}
	\end{figure}
	
		When the privacy parameter is set to $\epsilon=5$, we compare the convergence of each node. Figure \ref{fig2} depicts the variation in the regret of each node. As shown in Figure \ref{fig2}, all nodes converge, and the convergence speeds are almost identical. This indicates that through interaction with neighboring nodes, the regret of each node achieves sublinear convergence.
	
	\begin{figure}[htbp]
		\centering
		\includegraphics[scale=0.6]{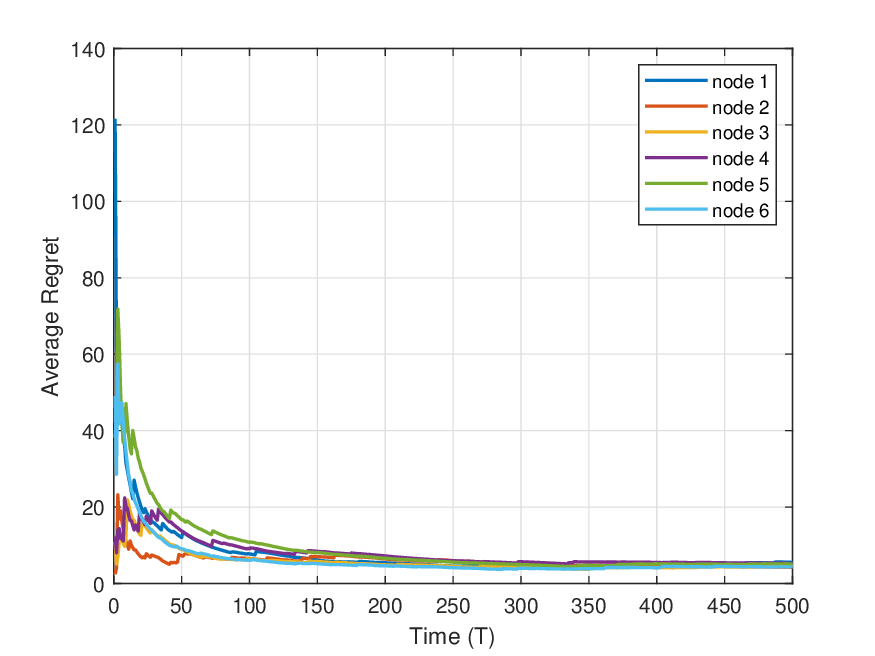}
		\caption{Comparison of each node's regret when $\epsilon=5$.}
		\label{fig2}
	\end{figure}
	
	We also present the convergence of the algorithm with different Bregman divergences. We demonstrate the algorithm's performance using the squared Euclidean distance and the Mahalanobis distance, with $\varphi = \|x\|^2$ and $\varphi = x^T Qx$, where $Q$ is a positive semi-definite matrix. Both distance functions satisfy Assumption \ref{AS4}. Figure \ref{fig3} shows the evolution of $\max_i \mathbb{E}[R_T^{i}]/T$ under different Bregman divergence with $\epsilon=5$. As seen in Figure \ref{fig3}, for different Bregman divergences, as long as they satisfy Assumption \ref{AS4}, the proposed differential privacy distributed online nonconvex optimization algorithm can converge to the stationary point, and the regret is sublinear.
	
	\begin{figure}[htbp]
		\centering
		\includegraphics[scale=0.6]{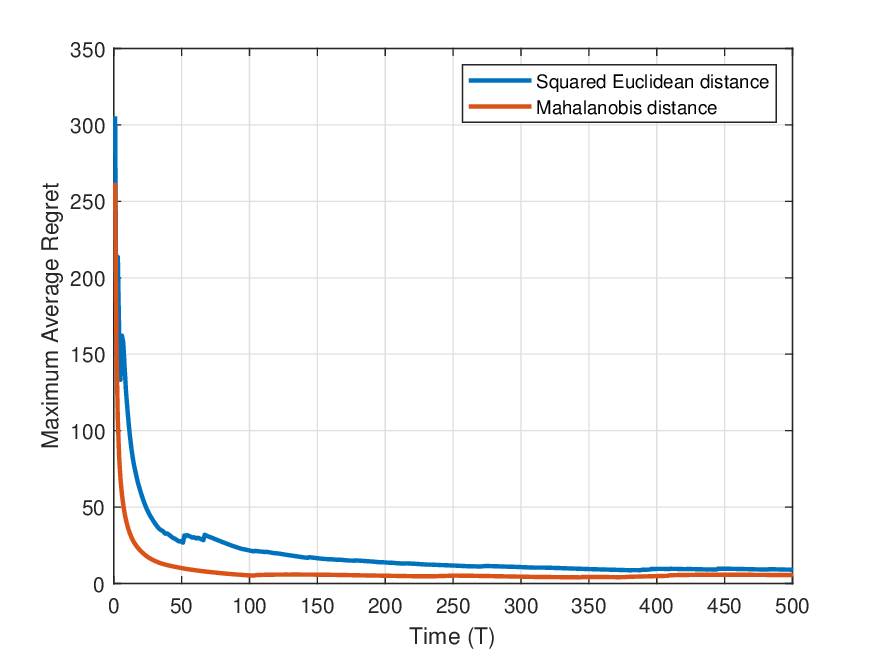}
		\caption{ Evolution of $\max_i \mathbb{E}[R_T^{i}]/T$ with different Bregman divergences.}
		\label{fig3}
	\end{figure}
	
	\section{Conclusion}
	In this paper, we study the privacy-preserving distributed online optimization for nonconvex problems over time-varying graphs. Based on the Laplace differential privacy mechanism and the distributed mirror descent algorithm, we propose a  privacy-preserving distributed online mirror descent algorithm for nonconvex optimization (DPDO-NC), which guarantees $\epsilon$-differential privacy at each iteration.
	In addition, we establish the upper bound of the regret for the proposed algorithm
	and prove that the individual regret based on the first-order optimality condition grows sublinearly, with an upper bound $O(\sqrt{T})$.
	Finally, we demonstrate the effectiveness of the algorithm through numerical simulations and analyze the relationship between the algorithm regret and the privacy level. Future research could explore more complex communication networks and scenarios with intricate constraints, as well as focus on improving the convergence rate of algorithms for solving distributed online nonconvex optimization problems.
	
	\section*{Acknowledgements}
	This work was supported  by the National Natural Science Foundation
	of China under Grant 62261136550.

		\bibliographystyle{elsarticle-num}
		\bibliography{ZL2024SCL}

\begin{thebibliography}{10}
\expandafter\ifx\csname url\endcsname\relax
  \def\url#1{\texttt{#1}}\fi
\expandafter\ifx\csname urlprefix\endcsname\relax\def\urlprefix{URL }\fi
\expandafter\ifx\csname href\endcsname\relax
  \def\href#1#2{#2} \def\path#1{#1}\fi

\bibitem{application_signal1}
G.~Tychogiorgos, A.~Gkelias, K.~K. Leung, A non-convex distributed optimization
  framework and its application to wireless ad-hoc networks, IEEE Transactions
  on Wireless Communications 12~(9) (2013) 4286--4296.

\bibitem{application_energy1}
W.~Chen, T.~Li, Distributed economic dispatch for energy internet based on
  multiagent consensus control, IEEE Transactions on automatic control 66~(1)
  (2020) 137--152.

\bibitem{application_decision}
W.~Ma, J.~Wang, V.~Gupta, C.~Chen, Distributed energy management for networked
  microgrids using online admm with regret, IEEE Transactions on Smart Grid
  9~(2) (2016) 847--856.

\bibitem{optimal_regret_bound}
P.~Nazari, D.~A. Tarzanagh, G.~Michailidis, Dadam: A consensus-based
  distributed adaptive gradient method for online optimization, IEEE
  Transactions on Signal Processing 70 (2022) 6065--6079.

\bibitem{bestregret}
E.~Hazan, A.~Agarwal, S.~Kale, Logarithmic regret algorithms for online convex
  optimization, Machine Learning 69~(2) (2007) 169--192.

\bibitem{online_ADMM}
M.~Akbari, B.~Gharesifard, T.~Linder, Individual regret bounds for the
  distributed online alternating direction method of multipliers, IEEE
  Transactions on Automatic Control 64~(4) (2018) 1746--1752.

\bibitem{distributed_online1}
R.~Dixit, A.~S. Bedi, K.~Rajawat, Online learning over dynamic graphs via
  distributed proximal gradient algorithm, IEEE Transactions on Automatic
  Control 66~(11) (2020) 5065--5079.

\bibitem{time-varying-graph}
J.~Li, C.~Gu, Z.~Wu, Online distributed stochastic learning algorithm for
  convex optimization in time-varying directed networks, Neurocomputing 416
  (2020) 85--94.

\bibitem{pseudoconvex}
K.~Lu, G.~Jing, L.~Wang, Online distributed optimization with strongly
  pseudoconvex-sum cost functions, IEEE Transactions on Automatic Control
  65~(1) (2019) 426--433.

\bibitem{beck2003mirror}
A.~Beck, M.~Teboulle, Mirror descent and nonlinear projected subgradient
  methods for convex optimization, Operations Research Letters 31~(3) (2003)
  167--175.

\bibitem{ben2001ordered}
A.~Ben-Tal, T.~Margalit, A.~Nemirovski, The ordered subsets mirror descent
  optimization method with applications to tomography, SIAM Journal on
  Optimization 12~(1) (2001) 79--108.

\bibitem{distributed_mirror_descent1}
S.~Shahrampour, A.~Jadbabaie, Distributed online optimization in dynamic
  environments using mirror descent, IEEE Transactions on Automatic Control
  63~(3) (2017) 714--725.

\bibitem{distributed_mirror_descent2}
D.~Yuan, Y.~Hong, D.~W. Ho, S.~Xu, Distributed mirror descent for online
  composite optimization, IEEE Transactions on Automatic Control 66~(2) (2020)
  714--729.

\bibitem{differential_privacy}
C.~Dwork, F.~McSherry, K.~Nissim, A.~Smith, Calibrating noise to sensitivity in
  private data analysis, in: Theory of Cryptography: Third Theory of
  Cryptography Conference, TCC 2006, New York, NY, USA, March 4-7, 2006.
  Proceedings 3, Springer, 2006, pp. 265--284.

\bibitem{hongyiguang_stronglyconvex}
M.~Yuan, J.~Lei, Y.~Hong, Differentially private distributed online mirror
  descent algorithm, Neurocomputing 551 (2023) 126531.

\bibitem{sconvex_umbalanced_graph}
Z.~Zhao, Z.~Yang, M.~Wei, Q.~Ji, Privacy preserving distributed online
  projected residual feedback optimization over unbalanced directed graphs,
  Journal of the Franklin Institute 360~(18) (2023) 14823--14840.

\bibitem{convex_directed_graph}
Q.~L{\"u}, K.~Zhang, S.~Deng, Y.~Li, H.~Li, S.~Gao, Y.~Chen, Privacy-preserving
  decentralized dual averaging for online optimization over directed networks,
  IEEE Transactions on Industrial Cyber-Physical Systems 1 (2023) 79--91.

\bibitem{time-varying_differential_privacy}
C.~Li, P.~Zhou, L.~Xiong, Q.~Wang, T.~Wang, Differentially private distributed
  online learning, IEEE transactions on knowledge and data engineering 30~(8)
  (2018) 1440--1453.

\bibitem{stronglyconvex}
J.~Zhu, C.~Xu, J.~Guan, D.~O. Wu, Differentially private distributed online
  algorithms over time-varying directed networks, IEEE Transactions on Signal
  and Information Processing over Networks 4~(1) (2018) 4--17.

\bibitem{online_defferential_privacy_extra1}
Y.~Xiong, J.~Xu, K.~You, J.~Liu, L.~Wu, Privacy-preserving distributed online
  optimization over unbalanced digraphs via subgradient rescaling, IEEE
  Transactions on Control of Network Systems 7~(3) (2020) 1366--1378.

\bibitem{online_defferential_privacy_extra2}
H.~Wang, K.~Liu, D.~Han, S.~Chai, Y.~Xia, Privacy-preserving distributed online
  stochastic optimization with time-varying distributions, IEEE Transactions on
  Control of Network Systems 10~(2) (2022) 1069--1082.

\bibitem{nonconvex_example_1}
J.~Zhang, S.~Ge, T.~Chang, Z.~Luo, Decentralized non-convex learning with
  linearly coupled constraints: Algorithm designs and application to vertical
  learning problem, IEEE Transactions on Signal Processing 70 (2022)
  3312--3327.

\bibitem{nonconvex_example_2}
S.~Hashempour, A.~A. Suratgar, A.~Afshar, Distributed nonconvex optimization
  for energy efficiency in mobile ad hoc networks, IEEE Systems Journal 15~(4)
  (2021) 5683--5693.

\bibitem{Hazan_nonconvex}
E.~Hazan, K.~Singh, C.~Zhang, Efficient regret minimization in non-convex
  games, in: International Conference on Machine Learning, PMLR, 2017, pp.
  1433--1441.

\bibitem{online_Newton}
A.~Lesage-Landry, J.~A. Taylor, I.~Shames, Second-order online nonconvex
  optimization, IEEE Transactions on Automatic Control 66~(10) (2020)
  4866--4872.

\bibitem{online_nonconvex}
J.~Li, C.~Li, J.~Fan, T.~Huang, Online distributed stochastic gradient
  algorithm for non-convex optimization with compressed communication, IEEE
  Transactions on Automatic Control 69~(2) (2024) 936--951.

\bibitem{lu_nonconvex1}
K.~Lu, L.~Wang, Online distributed optimization with nonconvex objective
  functions: Sublinearity of first-order optimality condition-based regret,
  IEEE Transactions on Automatic Control 67~(6) (2021) 3029--3035.

\bibitem{wang2023decentralized}
Y.~Wang, T.~Ba{\c{s}}ar, Decentralized nonconvex optimization with guaranteed
  privacy and accuracy, Automatica 150 (2023) 110858.

\bibitem{khajenejad2022guaranteed}
M.~Khajenejad, S.~Mart{\'\i}nez, Guaranteed privacy of distributed nonconvex
  optimization via mixed-monotone functional perturbations, IEEE Control
  Systems Letters 7 (2022) 1081--1086.

\bibitem{stochastic_matrix}
A.~Nedic, A.~Ozdaglar, Distributed subgradient methods for multi-agent
  optimization, IEEE Transactions on Automatic Control 54~(1) (2009) 48--61.

\bibitem{example11}
M.~Cao, B.~D. Anderson, A.~S. Morse, Sensor network localization with imprecise
  distances, Systems \& control letters 55~(11) (2006) 887--893.

\bibitem{hua2024distributed}
Y.~Hua, S.~Liu, Y.~Hong, K.~H. Johansson, G.~Wang, Distributed online bandit
  nonconvex optimization with one-point residual feedback via dynamic regret,
  arXiv preprint arXiv:2409.15680.

\bibitem{privacy_definition}
C.~Dwork, Differential privacy, in: International colloquium on automata,
  languages, and programming, Springer, 2006, pp. 1--12.

\bibitem{DP-definition}
Z.~Huang, S.~Mitra, N.~Vaidya, Differentially private distributed optimization,
  in: Proceedings of the 16th International Conference on Distributed Computing
  and Networking, 2015, pp. 1--10.

\bibitem{DP-definition2}
Y.~Wang, A.~Nedi{\'c}, Tailoring gradient methods for differentially private
  distributed optimization, IEEE Transactions on Automatic Control 69~(2)
  (2023) 872--887.

\bibitem{mirror_descent}
A.~Beck, M.~Teboulle, Mirror descent and nonlinear projected subgradient
  methods for convex optimization, Operations Research Letters 31~(3) (2003)
  167--175.

\end{thebibliography}
	\end{document}